\documentclass[runningheads]{llncs}
\usepackage{graphicx}
\usepackage{amsmath,amssymb} 
\usepackage{xypic,rotating}
\usepackage[width=122mm,left=12mm,paperwidth=146mm,height=193mm,top=12mm,paperheight=217mm]{geometry}

\xyoption{all}

\newcommand{\dqt}[1]        {``{#1}''}
\newcommand{\separate}    {\vspace{0cm}\begin{center}*~~~~~~~~~~*~~~~~~~~~~*\end{center}\vspace{0cm}}

\newcommand{\tset}[1]      {\{{#1}\}}                     

\newcommand{\cb}          {\begin{tabbing}MMMMM\=MM\=MM\=MM\=MM\=MM\=MM\=MM\=MM\=MM\= \kill}
\newcommand{\ce}          {\end{tabbing}}

\newcommand{\dom}[1]      {\mbox{Dom}({#1})}
\newcommand{\rg}[1]      {\mbox{Rg}({#1})}
\newcommand{\tpset}      {{\mathcal{T}}}

\newcommand{\tmach}      {\phi}
\newcommand{\tape}       {\tau}
\newcommand{\comp}       {{\mathfrak{M}}}
\newcommand{\rcomp}[2]   {{\mathfrak{M}}({#1},{#2})}
\newcommand{\ocomp}[1]   {{\mathfrak{M}}({#1})}
\newcommand{\acomp}[2]   {{\mathfrak{M}}^{{#1}}({#2})}

\newcommand{\carrow}     {\twoheadrightarrow}          
\newcommand{\tmstop}     {\!\downarrow} 
\newcommand{\tmnostop}   {\!\uparrow} 
\newcommand{\repa}       {{\mathfrak{a}}}
\newcommand{\repb}       {{\mathfrak{b}}}
\newcommand{\repc}       {{\mathfrak{c}}}
\newcommand{\repx}       {{\mathfrak{x}}}
\newcommand{\repy}       {{\mathfrak{y}}}
\newcommand{\repu}       {{\mathfrak{u}}}
\def\parf{\ar@{=}[d]}

\newcommand{\wleq}       {\preceq}
\newcommand{\weqv}       {\stackrel{c}{\sim}}
\newcommand{\wpar}       {||}
\newcommand{\sleq}       {\subseteq}
\newcommand{\seqv}       {\stackrel{t}{\sim}}
\newcommand{\spar}       {||}
\newcommand{\bij}[1]     {\langle{#1}\rangle}
\newcommand{\tdg}[2]     {\mbox{deg}_{#1}({#2})}
\newcommand{\tdgind}[1]   {\pmb{{#1}}}
\newcommand{\rdg}[1]     {\mbox{rdg}({#1})}
\newcommand{\aleq}[1]       {\underset{{#1}}{\subseteq}}
\newcommand{\aeqv}[1]       {\overset{c}{\underset{{#1}}{\sim}}}
\newcommand{\apar}[1]       {\underset{{#1}}{||}}
\newcommand{\aseqv}[1]		{\overset{t}{\underset{{#1}}{\sim}}}
\newcommand{\frep}[2]       {{#1}_{#2}}

\newcounter{qqremark}
\newcounter{qexample}
\newcounter{qproperty}
\def\example{
\bigskip

\refstepcounter{qexample}%
\noindent \textbf{Example \Roman{qexample}:}\\
}

\def\remark{
\refstepcounter{qqremark}
\noindent \emph{Remark \arabic{qqremark}:}
}

\def\property{
\refstepcounter{qproperty}
\noindent {Property \arabic{qproperty}:}
}

\begin{document}

\pagestyle{headings}
\mainmatter

\title{On the relation between representations and computability} 

\titlerunning{Representation and Computability}

\authorrunning{Casanova, Santini}

\author{Juan Casanova$^1$ \and Simone Santini$^3$\thanks{The work was
    carried out when Juan Casanova was with the Escuela Polit\'ecnica
    Superior, Universidad Aut\'onoma de Madrid; Simone Santini was
    supported in part by the the Spanish \emph{Ministerio de
      Educaci\'on y Ciencia} under the grant N. TIN2016-80630-P,
    \emph{Recomendaci\'on en medios sociales: contexto, diversidad y
      sesgo algor{\'\i}tmico.}  
  } 
}
\institute{Centre for Intelligent Systems and their Applications,
  School of Informatics, University of Edinburgh, UK
 \and 
Departamento de Ingenier{\'\i}a Inform\'atica,  Escuela Polit\'ecnica Superior, Universidad Aut\'onoma de Madrid,
 Spain}

\maketitle

\begin{abstract}
  One of the fundamental results in computability is the existence of
  well-defined functions that cannot be computed. In this paper we
  study the effects of data \emph{representation} on computability; we
  show that, while for each possible way of representing data there
  exist incomputable functions, the computability of a specific
  abstract function is never an absolute property, but depends on the
  representation used for the function domain. We examine the scope of
  this dependency and provide mathematical criteria to favour some
  representations over others.  As we shall show, there are strong
  reasons to suggest that computational enumerability should be an
  additional axiom for computation models.

  We analyze the link between the techniques and effects of
  representation changes and those of oracle machines, showing an
  important connection between their hierarchies. Finally, these
  notions enable us to gain a new insight on the Church-Turing thesis:
  its interpretation as the underlying algebraic structure to which
  computation is invariant.
 \end{abstract}

\section{Introduction}
Both historically and conceptually, computability and decidability
arise as a result of trying to formally deal with the way work in
mathematics itself is performed. In an attempt to provide a
formalization for the work of a mathematician proving theorems and
deriving results, several authors proposed respective models, such as
the Turing Machine \cite{turing:36} or Church's $\lambda$-calculus
\cite{church:36}. The arguably most important result deriving from
this work is the proof that there are mathematically well defined
functions which cannot be computed by any of these systems. The
halting theorem \cite{turing:36} is one of the most well known forms
of this statement, particularly proving the incomputability of the
so-called halting problem. In close relation to this statement is
G\"odel's incompleteness theorem \cite{godel:31} proving that no
formal system capable of expressing basic arithmetic can be both
consistent and complete, a more general proof of the same underlying
fact that the halting theorem proves.

The classic concept of \dqt{degree of recursive unsolvability}
introduced by Post \cite{post:44} appears as a result of extending the
concept of reducibility, implicit in Turing's work \cite{turing:39},
to incomputable functions. Formally, an oracle machine is a
theoretical device that enables a Turing machine to compute an
otherwise incomputable function (respectively, set) at any step of its
computation, thus making a whole new set of previously incomputable
functions (r., sets) computable under the new model, while leaving
other functions (r., sets) incomputable. This introduces a preorder
relation among functions (r., sets) that enables the definition of
equivalence classes (the degrees of recursive unsolvability) and a
partial ordering among them \cite{kleene:52}. The structure of this
set has been the subject of intensive study, much of it related to the
properties of the \emph{Turing jump}. The Turing jump is the extension
of the basic halting problem to Turing machines with oracle, providing
an effective method for obtaining an incomputable function (r., set)
for any model of computation defined through a Turing machine with
oracle. The classic 1954 paper by Kleene and Post \cite{kleene:54}
contains the fundamental results in this area: the structure of the
Turing degrees is that of a join-semilattice, for any set $A$, there
are $\aleph_0$ degrees between $\pmb{a}$ and $\pmb{a}'$. The global
structure of the set of degrees of recursive unsolvability,
$\mathcal{D}$, has been the object of a significant amount of work,
especially in the period 1960--1990
\cite{posner:81,sacks:61,sacks:64,slaman:01}. Notable results include
the fact that every finite lattice is isomorphic to an initial segment
of ${\mathcal{D}}$ \cite{lerman:71} and that every finite Boolean
algebra \cite{rosenstein:68} and every countable linear order with
least element \cite{hugill:69} can be embedded in an initial segment
of ${\mathcal{D}}$. The strongest result in this area was obtained in
\cite{abraham:86}: every initial element of an upper semi-lattice of
size $\aleph_1$ with the countable predecessor property occurs as an
initial segment of ${\mathcal{D}}$.

A different set of questions arise if we pose some restrictions on
the general Turing reducibility. By limiting the access to the oracle
one obtains several types of \emph{strong reducibilities}
\cite{odifreddi:99}, among which the most significant are the
truth-table \cite{feferman:57} and enumeration \cite{friederberg:59}
reducibilities. 

However, while representation is an arguably omnipresent concept in
all kinds of computability and computer science work, little specific
research has been carried about its effects in computability of
functions (respectively, sets) from a purely abstract and
formalism-centred point of view. In this article, we examine this and
related ideas, motivated by a small set of examples, introduced at the
beginning of the text, that show that, without further formal
constraints, the choice of the way data is represented in any
computation formalism is in principle a trascendental decision that
greatly affects the notion of computability.

Following the examples, we provide a short and straight to the point
representation-free definition of the Turing machine formalism which
allows us to restate some of the most common results in computability
theory in absolutely formal terms, with no dependency on the
representation chosen. The extremely formal character of these results
makes them insufficient for answering the common questions in
computability, and thus, in the following section we introduce a
formal definition of representation along with some basic derived
concepts and results.

We then define two related but not identical partial ordering structures
between representations (transformability and computational strength),
which enable us to prove the two main results of the article. These
results trace the boundaries on the effect that representation may
have on computability and provide insight on the essential elements of
computability that are in closest relation with the admittedly odd
effects of representation changes shown in the introductory examples
(namely, computational enumerability of sets and domain restriction of
functions). From these results we conclude that some of the most
counter-intuitive and consistency-defying issues raised by
representation changes can be dealt with by restricting the definition
of computability, however leaving some of the important effects of
representation on computability inevitably remaining.

Finally, motivated by the resemblance of the effects of
representations to those of oracle machines and of the relations
between representations to the hierarchy of degrees of recursive
unsolvability, we study the relation between these two concepts in two
different ways. First, we show that the hierarchy of degrees of
recursive unsolvability does not correspond to the hierarchy of the
so-called representation degrees, but have a strong connection with
them. Second, we show that all the definitions and results regarding
representations introduced in this article can be naturally and easily
extended to computability with oracle machines.

As a last consideration, we examine the relation that the
Church-Turing thesis and related concepts have with the ideas
introduced throughout the article.

\section{Some definitions}
Let ${\mathcal{F}}$ be the set of all functions defined on countable
sets.  Given a function $f$, we indicate as usual with $\dom{f}$ the
set of values for which $f$ is defined and with $f(\dom{f})=\rg{f}$
the set of values that are the image of elements of $\dom{f}$.  The
restriction $f_{|A}$ of $f$ to $A\subseteq\dom{f}$ is the function
defined on $A$ such that for all $x\in{A}$, $f_{|A}(x)=f(x)$.

The set of functions $(A\rightarrow{B})$ is defined in the usual way
\begin{equation}
  (A\rightarrow B) = \tset{f \in {\mathcal{F}}\,\big|\,A\subseteq\dom{f} \wedge \rg{f}\subseteq{B}}
\end{equation}
We shall use a square bracket to indicate that a set coincides with
the domain or range of a function, so we shall define
\begin{equation}
  [A\rightarrow B] = \tset{f \in {\mathcal{F}}\,\big|\,A=\dom{f} \wedge \rg{f}={B}}
\end{equation}
with the obvious meaning for $[A\rightarrow{B})$ and
  $(A\rightarrow{B}]$. Note that, of course,
$[A\rightarrow{B}]\subseteq(A\rightarrow{B})$.

\separate

In order to ground the intuition that drives this article, namely that
representations bear some relevance for computability, we shall
present here two examples, which we shall use throughout the paper.

\example%
\label{excomparator}
Consider the function $eq \in [\mathbb{N} \times \mathbb{N}
  \rightarrow \{0,1\}]$ that compares its two argumens, defined as
$eq(m,n)=\mbox{if~}m=n\mbox{~then~}1\mbox{~else~}0$.

Consider now two representations of pairs of natural numbers: in the
first representation, one of the numbers is coded to the left of the
initial head position (the least significant digit closest to the
initial head position), and the other number is coded to the right of
the initial head position, symmetrically, both in binary code, with no
additional information.  In the second representation, they are also
represented one in each direction from the initial head position, but
in unary code ($n$ is coded as $n+1$ \dqt{1}s followed by zeroes).

It is easy to see that in the first representation, $eq$ is not
computable.  For assume there exists a Turing machine $\phi$ which
succesfully compares two numbers in binary representation. Let $n$ be
any natural number. Run $\phi$ with the binary representation of
$(n,n)$. Since it computes $eq$, it must halt with a positive result
after a finite number of computation steps. Write $s_n \in \mathbb{N}$ for the
number of steps $\phi$ executes before halting with input $(n,n)$. Now
consider any natural number $m$ such that $m \neq n$ but $m = n$ mod
$2^{s_n}$. That is, $m$ and $n$ have the same $s_n$ least significant
binary digits. An infinite amount of such numbers exist. If we run
$\phi$ with $(m,n)$ as input, it will execute the first $s_n$ steps,
in which the head cannot go any further from the initial position than
$s_n$ cell positions. Thus, all cells inspected by $\phi$ during this
execution are equal to the ones inspected when running with input
$(n,n)$. Therefore, $\phi$ must necessarily halt after $s_n$ steps,
with a positive result. However, $m \neq n$. This proves that $\phi$
does not succesfully compare any two numbers.

Note that we are actually proving something stronger: no finite
computability model (a model that can't do an infinite number of
operations at the same time) can compare two numbers in finite time
using the first representation.

However, in the latter representation, the function is trivially
computable. Intuitively, the machine runs alternatively in both
directions until it encounters the last "1" on one side. If it is also
the last one on the other side, they are equal, otherwise, they are
distinct. Note that the problem with the former representation
disappears if the representation is slightly changed so that, as it
happens in all practical computers, the set of possible tape
configurations is finite or the boundaries of the tape are marked with
a special symbol.

\example%
\label{exhalting}
Consider the halting problem: find a Turing machine $\phi_H$ such
that, given any other Turing machine $\phi$ and an input tape $\tau$,
$\phi_H$ always halts and indicates whether $\phi$ with input tape
$\tau$ halts or not. One of the fundamental and best known textbook
theorems in computing science says that no such TM exist.

Consider the typical representation for which the halting
theorem is proved: The TM is represented on one side of the tape,
codified as a set of quintuples depicting its transition function, and
the input tape is on the other side, folded so that it can be
represented on just one side of the global tape.  In this
representation, the halting problem is undecidable.

However, the halting problem, in its most general form, is an abstract
problem about the properties of certain sets of Turing machines and input tapes.  In
order to consider its computability, we had to indicate that a Turing
machine is represented by the quintuples of its transition
function. But a TM is an abstract mathematical object that can be
codified on a tape in different ways. For example, we can extend the
previous representation by placing an additional symbol at the initial
head position. This symbol will be $1$ if the input Turing machine
halts with the given input tape, and $0$ otherwise.

In this representation, the TM $\phi_H$ exists trivially. Therefore
the halting problem is decidable.

\separate

\remark%
\label{compuremark}
One possible objection to the previous example is that we have
obtained a solution to the halting problem only because we have a
\dqt{non-computable} representation, that is, a representation in
which a non-computable quantity is computed and made explicit.

We must stress, however, that computability is defined only for
quantities represented on tape, that is, only once the abstract
entities involved have undergone the process of representation and its
computability is thus a strictly formulable fact in the Turing machine
formalism (or any other computation formalism). To say that a
representation is non-computable (or, for that matter, that it is is
computable) is meaningless, as representation is the prerequisite
necessary so that one can define computability.

\bigskip

The purpose of this article is to study the mathematical aspects of
representation and their relation to computability, in the light of
the two examples above and related facts. We will offer a theoretical
justification as to why the first representation offered for the
halting problem is more \dqt{reasonable} than the second. This
justification will allow us to acquire an insight of the mathematical
properties of computation formalisms and to better understand how
representations affect the notion of computability.

\section{Definitions and formalism}
Throughout most of this paper we shall consider representations in the
context of computation with Turing Machines, as it is arguably the
best known computing model, and one in which the explicit separation
between the computing device (the machine) and the input data (the
tape) makes the representation problem clearer and easier to work on.
Towards the end of the paper we will offer a discussion about
how the choice of formalism affects the discussions and conclusions
offered here and discuss briefly how formalism is related to a choice
of representation.

Many interesting results in computability can be formulated
meaningfully as statements on sets of tapes or strings, without the
need for representations. We will discuss these results later in this
section. Most of the results of this section are very well known facts
whose proofs can be found in textbooks. We shall therefore simply
remind the results, skipping the proofs.

We shall indicate with $\tpset_\Sigma$ the set of tapes on a finite
alphabet $\Sigma$ with a finite number of non-blank symbols; whenever
the alphabet in question is clear from the context, we shall omit the
subscript $\Sigma$; tapes will be indicated as $\tape$, $\tape'$ (or
$\tape_1$), etc.
As usual, given a tape $\tape$, we shall indicate with
$\tmach(\tape)\tmstop$ the fact that the machine $\tmach$ stops on input
$\tape$ and with $\tmach(\tape)\tmnostop$ the fact that $\tmach$ doesn't
stop.

Functions from tapes to tapes are especially interesting here as they
are the only ones for which Turing computability can be defined, as
observed in Remark \ref{compuremark}. A Turing machine, \emph{qua}
machine can take any tape as input and do something with it (possibly
never stopping, of course) but in general a function
$f\in(\tpset\rightarrow\tpset)$ will be defined only for a subset of
tapes. Therefore, a Turing machine \emph{qua} implementation of a
function $f$ is defined only in the domain of that function. 

Representations often come with a restriction in the set of
\dqt{valid} tapes. In example \ref{excomparator}, any tape can be
interpreted as the binary representation of two numbers, but only
tapes consisting in a collection of consecutive \dqt{1}, with \dqt{0}
everywhere else can be interpreted as the unary representation of
numbers. A TM $\phi_{\mbox{eq}}$ that implements a comparator on this
representation would of course work on any tape, in the sense that
whatever may be the input tape the TM $\phi_{\mbox{eq}}$ would operate,
but on these tapes $\phi_{\mbox{eq}}$ would not be an implementation
of the comparator function.

\begin{definition}
 Let $Q,P \subseteq\tpset$; a function $f \in[Q\rightarrow{P})$ is
  \emph{computable} if there is a TM $\tmach$ such that
    \begin{equation}
      \begin{aligned}
        \tape \in Q &\Rightarrow \tmach(\tape) = f(\tape)
      \end{aligned}
    \end{equation}
\end{definition}

\begin{definition}
  Let $Q,P\subseteq\tpset$; a function $f\in[Q\rightarrow{P})$ is
    \emph{partially computable} (p.c.) if it is computable and
    \begin{equation}
      \begin{aligned}
        \tape \not\in Q &\Rightarrow \tmach(\tape)\tmnostop
      \end{aligned}
    \end{equation}
\end{definition}

\begin{definition}
  A function $f\in[Q\rightarrow{P})$ is \emph{total computable} (t.c.) if it is
    p.c. and $Q=\tpset$.
\end{definition}

We shall indicate with $\comp$ the set of computable functions, and
with a doubly pointed arrow the fact that a specific function is
computable, that is, $f\in[Q\carrow{P}]$ entails that $f$ is
computable, that is:
\begin{equation}
  [Q\carrow P] = \comp \cap [Q \rightarrow P]
\end{equation}
whenever such an arrow appears in a diagram, the diagram will be said
to commute if, for each doubly pointed arrow there is a computable
function that makes the diagram commute in the traditional sense.  The
following theorem is the functional formulation of the standard result
on the composability of Turing machines.

\begin{theorem}
  \label{compotheorem}
  Let $f\in[Q\carrow{P})$ and $g\in(P\carrow{R}]$,
  $P'=\rg{f}\cap\dom{g}$, $Q'=f^{-1}(P')$ and $R'=g(P')$, then the
  restriction of $g\circ{f}$ to $Q'$ is computable:
  \begin{equation}
    (g\circ{f}) \in [Q'\carrow R']
  \end{equation}
\end{theorem}

Let $\nu,o\in\tpset$ be any two tapes with $\nu\ne{o}$; we shall call these
the \emph{yes} and \emph{no} tapes.

\begin{definition}
  The \emph{characteristic function} $\chi_A\in[\tpset\rightarrow\tpset)$ of a set $A\subseteq\tpset$ is
  the function
  \begin{equation}
    \chi_A(\tape) = 
    \begin{cases}
      \nu & \mbox{if $\tape\in{A}$} \\
      o   & \mbox{if $\tape\not\in{A}$}
    \end{cases}
  \end{equation}
\end{definition}

\begin{definition}
  The set $A\subseteq\tpset$ is computable if $\chi_A$ is computable.
\end{definition}

\begin{definition}
  \label{recena}
  A set of tapes $A\subseteq\tpset$ is computationally-enumerable-A
  (c.e.-A) if the restriction of $\chi_A$ to $A$,
  $\chi_{A|A}\in[A\rightarrow\tset{\nu}]$ is p.c., that is, if there
  is a TM $\tmach$ such that
  \begin{equation}
    \begin{aligned}
      \tape \in A     &\Rightarrow \tmach(\tape) = \nu \\
      \tape \not\in A &\Rightarrow \tmach(\tape)\tmnostop
    \end{aligned}
  \end{equation}
\end{definition}

\begin{definition}
  \label{recenb}
  A set $A\subseteq\tpset$ is computationally-enumerable-B (c.e.-B) if
  there is a partially computable function $f\in[A\rightarrow{A}]$ and
  a tape $\tape_0\in{A}$ such that for each $\tape\in{A}$ there is
  $i\in{\mathbb{N}}$ such that $\tape=f^{i}(\tape_0)$; the function
  $f$ is called the \emph{enumerator} of the set.
\end{definition}

\begin{theorem}
  \label{rectheorem}
  A set $A\subseteq\tpset$ is c.e.-A iff it is
  c.e.-B
\end{theorem}

Because of this theorem, we can call sets with these properties simply
\emph{computationally enumerable}, or \emph{c.e.}

\separate

So far, we have considered TMs working on arbitrary sets of
tapes. However, example \ref{excomparator} shows that we should exert
caution in choosing our sets of tapes, lest we be unable to compute
very fundamental operations, such as comparing the representations of
two numbers.

At this point, we are still considering tapes \emph{qua} tapes,
without assuming that they are the representation of anything else.
Even so, however, it seems obvious that there is a certain number of
properties that a \dqt{reasonable} set of tapes must satisfy if we
want to do some reasonable computation with them. If we do not limit
the set of tapes that we are considering, example \ref{excomparator}
shows that we can't even determine whether two parts of a tape are
equal: trying to do some meaningful computation in these circumstances
would be a daunting task, exacerbating the pointlessness of doing so.
It is beyond the scope of this paper to offer a theoretical
justification of the properties that we assume to hold for a set of
tapes, but some partial pragmatic justification can be derived from
the way in which we shall use these properties in the remainder of
this paper.

Some fundamental operations, such as comparison and copy, require the
capacity to represent pairs of tapes (or, more in general, tuples of
tapes) on a single tape. To this end, we consider a tape bijection:
\begin{equation}
  \bij{\_,\_} : \tpset\times\tpset\rightarrow\tpset
\end{equation}
such that the following operations are computable:
\label{bijfunctions}
\begin{tabbing}
MMMMMMMMMMMM\= \kill \\
Duplicate:  \> $\delta:\tape\mapsto\bij{\tape,\tape}$ \\
Swap:       \> $\sigma:\bij{\tape,\tape'}\mapsto\bij{\tape',\tape}$ \\
Projection: \> $\pi_1:\bij{\tape,\tape'}\mapsto\tape$ \\
Partial application: \> $\alpha_1[f]:\bij{\tape,\tape'}\mapsto\bij{f(\tape),\tape'}$ 
\end{tabbing}

Any two such bijections are computationally equivalent under the
definitions given in section \ref{relations}.

We assume that we can combine TMs in such a way that we can
detect when one of them enters an accepting state and continue the
computation consequently. In particular, we assume that we can define
a TM that recognizes whether its input is the same as a constant tape
$\tape$ (encoded in the stucture of the machine) and a TM (called
\emph{eq}) that, given a tape containing $\bij{\tape,\tape'}$ accepts
if $\tape=\tape'$ and rejects if $\tape\ne\tape'$.

The second projection $\pi_2$ and the partial application on the
second element of a bijection $\alpha_2(f)$ can be defined in terms of
the basic operations and the composition.

All the operations that we need in order to prove the results of this
paper can be defined in terms of these. The proof of this fact is easy
(one only has to show how to build the operations) but technical and we
omit it, as it is peripheral to the contents of the paper.

Under these assumptions, but conditioned to them, more of the typical
results can be proven.

\begin{theorem}
  \label{inverseth}
  If $f\in\comp$ and $\dom{f}$ is c.e. then $\rg{f}$ is
  c.e. Furthermore, if $f$ is injective, then
  $f^{-1}\in[\rg{f}\carrow\dom{f}]$ is computable.
\end{theorem}

Note that the function $f^{-1}$ is computable only on $\rg{f}$:
the theorem doesn't guarantee that this set be computable. \\

The set of tapes 
$\tpset$ is computationally enumerable. Given an enumerator
$E\in[\tpset\carrow\tpset]$ and an initial tape $\tape_0$, for a given
tape $\tape$, define $\#\tape=n$ ($n\in{\mathbb{N}}$) if
$\tape=E^n(\tape_0)$. As we will see in the following, this is
equivalent to considering a representation of natural numbers
which allows the computation of the successor function.

The set of TMs is also countable.  One can thus consider the $m$th
Turing machine $\tmach_m$, for $m\in{\mathbb{N}}$ under a certain
enumeration. In order to allow the use of a universal Turing Machine,
the enumeration that we use must allow us, given a tape with the
representation of a number $m$ on it, to execute $\phi_m$.  Unless we
state the contrary, the standard enumeration of Turing machines which
we will consider will have this universality property.

Given a TM $\tmach_e$ (under a particular enumeration of Turing
machines), $e\in{\mathbb{N}}$, define
\begin{equation}
  W_e = \tset{\tape\in\tpset | \tmach_e(\tape)\tmstop}
\end{equation}
We indicate with $W_{e,n}$, $e,n\in{\mathbb{N}}$ the set of
tapes that $\tmach_e$ accepts after $n$ steps.\\

\begin{theorem}
  A set $A\subseteq\tpset$ is c.e. iff $A=\emptyset$ or $A$ is the
  range of a computable function.\\
\end{theorem}

Finally, we give a version of the halting theorem which can be
formulated in exclusive terms of Turing machines.

\begin{theorem}
  \label{knotcomp}
	Let $K = \tset{\tape | \phi_{\#\tape}(\tape)\tmstop}$. $K$ is not computable
\end{theorem}

\begin{proof}
  Suppose $K$ has a computable characteristic function $\chi_K$; define
  \begin{equation}
    f(\tape) = 
    \begin{cases}
      E(\phi_{\#\tape}(\tape)) & \mbox{if $\chi_K(\tape)=\nu$} \\
      o                       & \mbox{if $\chi_K(\tape)=o$}
    \end{cases}
  \end{equation}
  then $f\in\comp$. Thus, we know that there exists $\tau_0$ such that
  for all $\tau \in \tpset$, $\phi_{\#\tau_0}(\tau)\tmstop$ and
  $f(\tau) = \phi_{\#\tau_0}(\tau)$ but, for all $\tape$ such that
  $\phi_{\#\tau}(\tau)\tmstop$, $f(\tape) = E(\phi_{\#\tape}(\tape))
  \ne\phi_{\#\tape}(\tape)$. In particular,
  $\phi_{\#\tau_0}(\tau_0)\tmstop$ and thus $f(\tau_0) \ne
  \phi_{\#\tau_0}(\tau_0)$. This contradiction proves that $K$ must be
  not computable.
\end{proof}

This theorem shows that not all undecidability can be whisked away
with a suitable choice of representation, the way we have done in
example \ref{exhalting}. Undecidable problems can be defined based
purely on tape computation, without assuming that the tapes are
representations of abstract sets.

\section{Representations}
Having established a computation formalism and some basic results, we
can now offer a formal definition of representation. 

\begin{definition}
  A \emph{representation} of an abstract set $A$ into the set of tapes
  $\tpset$ is an injective function $\repa\in(A\rightarrow\tpset)$
\end{definition}

The \emph{image} of $A$ under $\repa$ is the set of tapes
$\repa(A)\subseteq\tpset$. Note that if $\repa$ is a representation of
$A$ and $A'\subseteq{A}$, then $\repa$ is also a representation of
$A'$ but its properties as a representation of $A$ might be quite
different from its properties as a representation of $A'$. For
example, the set $\repa(A)$ might be c.e. while $\repa(A')$ may fail
to be. As we'll see in the following, this fact has quite far-reaching
consequences.

The set $\tpset$ is countable and, since representations are required
to be injective, the abstract set $A$ is also countable: we can't
represent any set with cardinality higher than $\aleph_0$.

\begin{definition}
  Given a function $f\in[A\rightarrow{B}]$ and two representations
  $\repa$ and $\repb$ of $A$ and $B$, respectively, a
  \emph{representation} of $f$ is the function
  $\frep{f}{\repa\repb}:[\repa(A)\rightarrow\repb(B)]$ such that the
  following diagram commutes
  \[
  \xymatrix{
    A \ar[r]^{f} \ar[d]_{\repa} & B \ar[d]^{\repb} \\
    \repa(A) \ar[r]_{\frep{f}{\repa\repb}} & \repb(B)
  }
  \]
  If $\frep{f}{\repa\repb}\in\comp$, then we say that $f$ is
  \emph{computable} in the pair of representations $\repa$,
  $\repb$. We shall indicate with $\rcomp{\repa}{\repb}$ the set of
  functions computable in $(\repa,\repb)$.
\end{definition}

\bigskip

\remark Note that even if $A=B$ we don't assume necessarily that
$\repa=\repb$: the same set can be represented in two different ways
as arguments and as result of the function. If $A = B$ and $\repa =
\repb$, we can abbreviate $\rcomp{\repa}{\repa}$ as $\comp(\repa)$.

On the other hand, we can assume that $A = B$ and use only one
representation. If $A \neq B$, and given representations $\repa$ and
$\repb$ we can always consider the set $C=A\sqcup{B}$ (disjoint union)
and a representation $\repc$ such that $\repc(a)=\repa(a)$ and
$\repc(b)=\repb(b)$, if $a\in{A}$ and $b\in{B}$. If for certain
$a\in{A}$ and $b\in{B}$, $\repa(a)=\repb(b)$, then we can consider a
partially disjoint union of $A$ and $B$ in which we consider that
$a=b$. Using just one set and one representation for input and output
becomes particularly natural when considering that most computations
are to be composed, and they cannot be composed if they are on
different sets. In other words, usual computation is defined for
functions $(\mathbb{N}\rightarrow\mathbb{N})$ which can be
chained.

\remark The representation of a function is always defined: $\repa$ is
injective in $\repa(A)$ and therefore invertible, so we have
$\frep{f}{\repa\repb}=\repb\circ{f}\circ\repa^{-1}$. The representation,
however, may fail to be computable.

Computability is not an immediately transferred property. Consider two
functions $f_1,f_2\in[A\rightarrow{B}]$ in $\rcomp{\repa}{\repb}$, and
assume that $f_1\in\rcomp{\repa'}{\repb'}$. This doesn't entail that
$f_2\in\rcomp{\repa'}{\repb'}$. To see this, consider examples
\ref{excomparator} and \ref{exhalting}. The identity function is
computable in both representations (it is computable in any
representation), but the comparator in example \ref{excomparator} and
the halting decision function in example \ref{exhalting} are computable
only in one of them.

\begin{lemma}
  The identity function of any set is computable in any representation.
\end{lemma}

\begin{proof}
  The Turing machine with just one accepting state computes the identity.
\end{proof}

\begin{lemma}
  \label{finitecomputability}
  Let $A$ be a finite set. Let $f \in [A \rightarrow A)$ be any
    function on this set. Then, $f$ is computable in any
    representation $a$ of $A$.
\end{lemma}

\begin{proof}
  For each $a \in A$, we can build a Turing machine $\tmach_a$ such
  that, given the representation of another element $\repa(a') \in
  \repa(A)$, checks whether $\repa(a) = \repa(a')$. Since $A$ is
  finite, we can build a Turing machine which executes all machines
  $\tmach_a$ on any input $a'$. Since $a' \in A$, then one and only
  one of those machines will have a positive result. If the Turing
  machine with positive result was $\tmach_a$, then $a' = a$ and thus
  $f(a') = f(a)$. Write $f(a)$ in the result.
\end{proof}

\begin{lemma}
  \label{latest}
  Let $f \in [A \rightarrow A)$ be a constant function. Then, $f$ is
    computable in any representation $\repa$ of $A$.\\
\end{lemma}

\begin{proof}
  Let $\{x\} = Rg(f)$. Consider $\tape = \repa(x)$. There exists a
  Turing machine which replaces any input tape for the constant tape
  $\tape$.
\end{proof}

\remark%
Lemma \ref{latest} requires the assumption that the comparator be
computable or at least that there exists a way to verify when all the
input has been read. Consider the binary representation considered in
example \ref{excomparator}. Replacing the input tape with a constant
tape is not a computable problem in such representation. Even if the
boundaries of the constant tape are previously know, the boundaries of
the input tape are not, and thus it is not possible to decide when all
of the input tape's extra data have been erased.

Lemma \ref{finitecomputability} does not require these assumptions,
however, as we are not implementing a general comparator but rather a
comparator for a constant tape, which can be implemented in any
representation.

\bigskip

With these definitions at hand, consider example \ref{exhalting}
again. The usual proof of the halting problem fails in this case
because the new representation doesn't necessarily allow us to carry
on the manipulations that the proof of the theorem requires. In
particular, given the pair $\bij{\phi,\tape}$, the theorem assumes, by
\emph{reductio ad absurdum}, the existence of a TM $\phi_h$ implementing a
function $h$ such that
\begin{equation}
  h(\bij{\phi,\tape}) = 
  \begin{cases}
    1 & \mbox{if $\phi(\tape)\tmstop$} \\
    0 & \mbox{if $\phi(\tape)\tmnostop$}
  \end{cases}
\end{equation}
The diagonalization argument then assumes that given a
tape $\tape$, we execute $h(\bij{\tape,\tape})$; the argument
therefore assumes that the representation used for the pair
$\bij{\phi,\tape}$ allows us the computation of the function
\begin{equation}
  \tape \mapsto \bij{\tape,\tape}
\end{equation}
but if we represent the pair as $\bij{\bij{\phi,\tape},b}$, where $b$
is the solution of the halting problem for $\phi$ and $\tape$, then
the function
\begin{equation}
  \tape \mapsto \bij{\bij{\tape,\tape}, b}
\end{equation}
is computable only by solving the halting problem, and the argument
becomes circular.

\bigskip

Most of the well known results in computability theory are expressed
in terms of sets of natural numbers. We are now in the condition to
express them again but, this time, relative to the representation
used. We shall consider here only the concepts that we shall use in
the following. We shall consider the set $A$ that we are representing
as fixed.

\begin{definition}
  The \emph{domain} of the TM $\tmach_e$ in a representation $\repx$
  is
  \begin{equation}
    W_e^\repx = \tset{ a | a \in A \wedge \tmach_e(\repx(a))\tmstop}
  \end{equation}
\end{definition}

\begin{definition}
  If $U\subseteq{A}$, and $\repa$ and $\repb$ are representations of
  $A$ and $B$, respectively; a function $f\in[U\rightarrow{B}]$ is
  p.c. in the representations $\repa$, $\repb$ if there is
  $e\in{\mathbb{N}}$ such that
  \begin{equation}
    \begin{aligned}
      x \in U &\Rightarrow \tmach_e(\repa(x)) = \repb(f(x)) \\
      x \not\in U &\Rightarrow \tmach_e(\repa(x))\tmnostop
    \end{aligned}
  \end{equation}
\end{definition}

Clearly, partially computable functions are computable functions.

\begin{definition}
  A function $f\in[A\rightarrow{B}]$ is total computable in the
  representations $\repa$, $\repb$ if it is p.c. and
  $\repa(A)=\tpset$.
\end{definition}

Computability and enumerability of sets are extended in the obvious
way through the computability of their characteristic
functions. However, note that given a set $A$ and a representation
$\repa$ of $A$, $\repa(A)$ need not be necessarily c.e.. This makes
some basic and intuitive results become false when extended, as they
rely on the enumerability of $\tpset$. For example, if a set $B
\subset A$ is computable under representation $\repa$ (that is, the
function $\chi_B \in [A \rightarrow \{\nu,o\} \subset A]$ is
computable under representation $\repa$); then $B$ need not be
necessarily computationally enumerable. This suggests that
enumerability might be (along with comparation, duplication and the
other aforementioned elemental functions) another reasonable
hypothesis to require for a representation. We will come back to this
idea further on.

\subsection{Relations between representations}
\label{relations}

\begin{definition}
  Let $\repx$ and $\repy$ be two representations of a set $A$: $\repx$
  is (computationally) transformable in $\repy$ ($\repy\wleq\repx$)
  if there is a function $f\in[\repx(A)\twoheadrightarrow\repy(A)]$
  such that $\repy=f\circ\repx$, viz. such that
  \begin{equation}
    \xymatrix{
      &    \repx(A) \ar@{->>}[dd]^f \\
      A \ar[ur]^\repx \ar[dr]_\repy & \\
      &    \repy(A)
    }
  \end{equation}
  commutes.
\end{definition}

\bigskip

\remark
If we drop the computability requirement, then the function $f$ always
exists and is unique in $[\repx(A)\rightarrow\repy(A)]$, since the
representations are one-to-one and onto. The previous definition,
therefore, is tantamount to requiring that
$\repy\circ\repx^{-1}\in\comp$.

\separate

\begin{definition}
  Two representations $\repx$ and $\repy$ are \emph{transformationally
    equivalent} ($\repx\seqv\repy$) if $\repx\wleq\repy$ and
  $\repy\wleq\repx$.
\end{definition}

\bigskip

\remark
Transformability is transitive, as can be seen from the commutativity
of the following diagram and by the fact that $\comp$ is closed under
composition.
\begin{equation}
  \xymatrix{
    & \repx(A)  \ar@{->>}[d]^f \ar@{->>}@/^2em/[dd]^{g\circ{f}} \\
    A \ar[ru]^{\repx} \ar[r]^{\repx'} \ar[rd]^{\repx''} & \repx'(A) \ar@{->>}[d]^g \\
    & \repx''(A) 
  }
\end{equation}
Therefore $\seqv$ is an equivalence relation. Moreover, $\wleq$
induces a partial ordering on the equivalence classes.

\begin{theorem}
  Let $\repx,\repy$ be representations of a set $A$. If
  $\repy\wleq\repx$ and $\repx(A)$ is c.e., then $\repy\seqv\repx$.
\end{theorem}

\begin{proof}
  Since $\repy\wleq\repx$, there is a transformation
  $f\in[\repx(A)\twoheadrightarrow\repy(A)]$ and since $\repx(A)$ is
  c.e., by theorem \ref{inverseth}, there is
  $f^{-1}\in[\repy(A)\twoheadrightarrow\repx(A)]$ (also
  computable). Therefore $\repx\wleq\repy$ and $\repx\seqv\repy$.
\end{proof}

This theorem gives us some indications on the hierarchy induced by the
relation $\wleq$; it tells us that all representations whose
range is c.e. are at the bottom of the hierarchy: they are
computationally transformable only in relations equivalent to
them. This is no longer the case if we drop the c.e. requirement.

While transformability is an easy to understand and sound concept, it
fails to express all the concerns on computability of representations
we are trying to consider here. We introduce here another relation
which is neither generally stronger nor weaker than the
transformability relation.

\begin{definition}
  Let $\repx$ and $\repy$ be two representations of a set $A$; $\repx$
  is \emph{computationally stronger} than $\repy$ (written
  $\repy\sleq\repx$) if
  $\ocomp{\repy}\subseteq\ocomp{\repx}$.
\end{definition}

That is, $\repx$ is computationally stronger than $\repy$ if all
functions computable in the representation $\repy$ are also computable
in the representation $\repx$. This induces, of course, another
equivalence relation and a partial ordering among equivalence classes.

\begin{definition}
  Two representations $\repx$ and $\repy$ are \emph{computationally equivalent}
  ($\repx\weqv\repy$) if $\repx\sleq\repy$ and
  $\repy\sleq\repx$.
\end{definition}

\begin{definition}
  Two representations $\repx$ and $\repy$ are \emph{incomparable}
  ($\repx\wpar\repy$) if neither $\repx\sleq\repy$ nor
  $\repy\sleq\repx$.
\end{definition}

\example
We will provide a new, more technical formulation of example
\ref{exhalting}, now that we have the necessary concepts.  Let $\repc$
be the standard representation of ${\mathbb{N}}$, in which a number
$n$ is represented by a string of $n+1$ symbols \dqt{1} followed by a
\dqt{0}. Build a representation $\repy$ as follows. Define the
function $f(n)$ as
\begin{equation}
  f(n) = 
  \begin{cases}
    1 & \mbox{if $\tmach_n(n)\tmstop$} \\
    0 & \mbox{if $\tmach_n(n)\tmnostop$} 
  \end{cases}
\end{equation}
Define 
\begin{equation}
  \repy(n) = \bij{\repc(n),\repc(f(n))}
\end{equation}
It is clear that $\repy({\mathbb{N}})$ is not c.e. and that
$\repc\wleq\repy$, since $\repc=\pi_1\circ\repy$ and $\pi_1$ is
computable in $\repy$. On the other hand, it is not $\repc\seqv\repy$,
as in $\repy$ one can compute $f(n)$ simply as
\begin{equation}
 \repy(n) \mapsto 
 \begin{cases}
   \bij{\repc(0), \repc(f(0))} & \mbox{if $\pi_2(\repy(n))=\repc(0)$} \\
   \bij{\repc(1), \repc(f(1))} & \mbox{if $\pi_2(\repy(n))=\repc(1)$}
 \end{cases}
\end{equation}
while in $\repc$ this is not possible (as proved by the undecidability
of the halting problem).

\separate

The previous example shows that the hierarchy established by the
relation $\wleq$ is not trivial.

\example
As we have already mentioned, this paper is framed mostly in the
context of the Turing model of computation, that is, using Turing
machines. However, the concept of representations and the issues that
derive from it are not limited to Turing Machines. Here we give an
example that hints at the generality of representation by discussing
briefly the halting problem in the context of Church's
$\lambda$-calculus.

While the Turing model has two elements (namely, the machine, which
implements the computing algorithm and the tape, which contains the
input and output data), $\lambda$-calculus defines a single concept:
the $\lambda$-funtion. Data must be represented as $\lambda$-functions
on which other $\lambda$-functions operate. Computation is done
through function application, and it stops (\emph{if} it stops) when
the expression is reduced to a normal form (see \cite{barendregt:84} for
details). An expression that cannot be reduced to normal form is in a
sense the equivalent of a Turing machine that never stops. The halting
problem in the context of $\lambda$-calculus can therefore be
expressed as follows:

\begin{quote}
  \emph{Is there a $\lambda$-function $H$ (the halt detector) such
    that, given any $\lambda$-function $L$ (algorithm) and another
    $\lambda$-function $I$ (input), $H$ detects whether $L$, when
    given $I$ as input, has a normal form (halts)?}
\end{quote}

Here, \dqt{detecting} means that $H$ gives two possible results, one
being the 'yes' answer (a $\lambda$-function which we will label
$\nu$) and the other the 'no' answer ($o$). The function $H$ must work
on two parameters, namely $L$ and $I$. These two parameters will be
expressed by considering the $\lambda$-function $\lambda{a.aLI}$ as
the \emph{representation} of the tuple $(L,I)$, which is then
evaluated by feeding it the function which must be evaluated on $L$
and $I$. That is, for any lambda function $F$,
$F(L)(I)=(\lambda{a.aLI})(F)$.

Under this representation, the halting problem is known to be
undecidable, a fact derived from the equivalence between
$\lambda$-calculus and Turing machines. Consider, however, an
alternative representation.

We shall represent the tuple $(L,I)$ as the lambda
function $\lambda{a.aLIh}$, where $h$ is $\nu$ if $L$ evaluated on
$I$ has normal form, and $o$ if it does not. This is, of course, the
same 'trick' used in the Turing machine formalism. 

In this representation, $H$ can easily be computed with the function
\[
\lambda m. m (\lambda x. \lambda y. \lambda z. z). 
\]

In particular,
\[
\begin{aligned}
  (\lambda m. m (\lambda x. \lambda y. \lambda z. z)) (\lambda a. a L I h) &= (\lambda a. a L I h) (\lambda x. \lambda y. \lambda z. z) \\
  & = (\lambda x. \lambda y. \lambda z. z) (L) (I) (h)\\
  & = h.
\end{aligned}
\]

\separate

The examples that we have presented so far show a common pattern: a
representation is improved not by a \emph{semantic} change that uses
additional properties of the abstract set hitherto not used, but by a
change internal to the space of representations that can be carried
out independently of the abstract set that we are representing (adding
delimiting characters, including the result of certain functions,
etc.). This pattern is general, and reveals a fundamental property of
representation improvement.

Let $\repa$, $\repa'$ be two representations of $A$. Let $f\in[A \rightarrow A]$ and
$\frep{f}{\repa}$, $\frep{f}{\repa'}$ be the two corresponding representations
of $f$. We then have
\[
\xymatrix{
  \repa'(A) \ar[r]^{\frep{f}{\repa'}} & \repa'(A) \\
  A \ar[u]^{\repa'} \ar[d]_{\repa} \ar[r]^{f} & A \ar[u]^{\repa'} \ar[d]_{\repa} \\
  \repa(A) \ar[r]^{\frep{f}{\repa}} & \repa(A) 
}
\]
The composition $\repx=\repa'\circ\repa^{-1}$ is an endorepresentation: a representation
of $\tpset$ into $\tpset$: $\frep{f}{\repx}$ is a representation
of $\frep{f}{\repa}$%
\footnote{Formally $\frep{f}{\repx}=\frep{f}{\repa'}$, a
  function in $a'(A)$. Nevertheless, we use two different
  symbols because the interpretation of the two functions is not the
  same: $\frep{f}{\repa'}$ is a representation of $f$, while
  $\frep{f}{\repx}$ is a representation of $\frep{f}{\repa}$.}%
. Note that if $f\in\ocomp{\repa}$ and $f\in\ocomp{\repa'}$, then
$\frep{f}{\repa}\in\ocomp{\repx}$.  The reverse is also
true: if $f\in\ocomp{\repa}$ and
$\frep{f}{\repa}\in\ocomp{\repx}$, then
$f\in\ocomp{\repa'}$.\\

\begin{theorem}
Let $\repx$ and $\repy$ be two representations of a set $A$. Then,
$\repx$ is computationally stronger than $\repy$ if and only if $\repx
\circ \repy^{-1} \in [\repy(A) \rightarrow \repx(A)]$, as an
endorepresentation of $\repy(A) \subset \tpset$ is computationally
stronger than the trivial representation of $\repy(A)$ (the identity
function $id_{\repy(A)}$).
\end{theorem}

\begin{proof}
Let $\repy \sleq \repx$. Let $f \in [\repy(A) \rightarrow \repy(A))$ such that $f \in \ocomp{id_{\repy(A)}} \subset \comp$. Consider the function $g = \repy^{-1} \circ f \circ \repy \in [A \rightarrow A)$, $f = g_\repy \in \comp$. Therefore, $g \in \ocomp{\repy}$. Since $\repy \sleq \repx$, then $g \in \ocomp{\repx}$, that is, $g_\repx = \repx \circ g \circ \repx^{-1} = \repx \circ \repy^{-1} \circ f \circ \repy \circ \repx^{-1} = f_{\repx \circ \repy^{-1}} \in \comp$, and $f \in \ocomp{\repx \circ \repy^{-1}}$.\\

Now let $id_{\repy(A)} \sleq \repx \circ \repy^{-1}$, and let $f \in [A \rightarrow A)$ such that $f \in \ocomp{\repy}$. Consider $g = f_\repy \in \comp \cap [\repy(A) \rightarrow \repy(A))$. Thus, $g \in \ocomp{id_{\repy(A)}}$, and since $id_{\repy(A)} \sleq \repx \circ \repy^{-1}$, then $g \in \ocomp{\repx \circ \repy^{-1}}$. Then, $g_{\repx \circ \repy^{-1}} = \repx \circ \repy^{-1} \circ g \circ \repy \circ \repx^{-1} = \repx \circ \repy^{-1} \circ \repy \circ f \circ \repy^{-1} \circ \repy \circ \repx^{-1} = \repx \circ f \circ \repx^{-1} = f_\repx \in \comp$, and $f \in \ocomp{\repx}$.

\end{proof}

\begin{corollary}
Let $\repx$ and $\repy$ be two representations of a set $A$. Then,
$\repx$ and $\repy$ are computationally equivalent if and only if
$\repx \circ \repy^{-1}$ and $\repy \circ \repx^{-1}$ are
computationally equivalent to $id_{\repy(A)}$ and $id_{\repx(A)}$ respectively.
\end{corollary}

\begin{proof}

$\repx$ and $\repy$ are computationally equivalent if and only if
  $\repx \sleq \repy$ and $\repy \sleq \repx$.

$\repx \sleq \repy$ if and only if $id_{\repx(A)} \sleq \repy \circ
  \repx^{-1}$ and $\repy \sleq \repx$ if and only if $id_{\repy(A)}
  \sleq \repx \circ \repy^{-1}$.

On the other hand, let $f \in [\repy(A) \rightarrow \repy(A))$, $f \in
  \ocomp{\repx \circ \repy^{-1}}$ if and only if $f_{\repx \circ
    \repy^{-1}} = \repx \circ \repy^{-1} \circ f \circ \repy \circ
  \repx^{-1} \in \comp$. Because $f_{\repx \circ \repy^{-1}} \in
       [\repx \circ \repy^{-1} (A) \rightarrow \repx \circ \repy^{-1}
         (A)) \subset [\repx(A) \rightarrow \repx(A))$, then this
           implies that $f_{\repx \circ \repy^{-1}} \in
           \ocomp{id_{\repx(A)}}$. Now, $\repx \sleq \repy$ if and
           only if $id_{\repx(A)} \sleq \repy \circ \repx^{-1}$, and
           then $f_{\repx \circ \repy^{-1}} \in \ocomp{\repy \circ
             \repx^{-1}}$, if and only if $(f_{\repx \circ
             \repy^{-1}})_{\repy \circ \repx^{-1}} = \repy \circ
           \repx^{-1} \circ f_{\repx \circ \repy^{-1}} \circ \repx
           \circ \repy^{-1} = \repy \circ \repx^{-1} \circ \repx \circ
           \repy^{-1} \circ f \circ \repy \circ \repx^{-1} \circ \repx
           \circ \repy^{-1} = f \in \comp$, and thus $f \in
           \ocomp{id_{\repy(A)}}$.\\

Symmetrically prove that for every $f \in [\repx(A) \rightarrow
  \repx(A))$, $f \in \ocomp{\repy \circ \repx^{-1}}$ implies $f \in
  \ocomp{id_{\repx(A)}}$.
\end{proof}

According to these theorems improvements in representation of any
abstract set come through changes internal to the tape representation,
a change in the way the tapes are written. Conversely, every change
that represents an improvement in the space of representation
corresponds to an improvement in the representation of the abstract
set.  This theorem allows us to push all considerations about
representation improvement to within the space of tapes,
independently of the properties of the set that we are representing.

\remark
It is important to note that all relations between representations are
relative to their use \emph{qua} representations of a specific
set. So, when we state that $id_{\repy(A)}\sleq\repx\circ\repy^{-1}$,
the relation holds only when the representations are intended as
representations of the set $\repy(A)$. If we consider them as
representation of a different set (for example of a set
$B\subset{A}$), the relation may fail to hold.

For example, in lemma \ref{finitecomputability} we proved that all
representations of finite sets are computationally equivalent. This,
and the fact that there exist endorepresentations of $\tpset$ which
are not computationally equivalent to the identity would seem to lead
us to a contradiction. We could, the argument would go, strictly
improve a representation of a finite set by using this improvement in
the set of tapes. However, these endorepresentations of $\tpset$ which
are not computationally equivalent to the identity as representations
of $\tpset$ are computationally equivalent to the identity as
representations of the finite set.

\separate

\section{Representations and the Turing hierarchy}

\subsection{Oracles}
Given a set of tapes $Q$, an \emph{oracle} for $Q$ is the
characteristic function for $Q$. A TM with oracle $Q$, $\tmach_e^Q$ is
a TM with one additional operation: given a tape $\tape$, the
operation produces $\nu$ if $\tape\in{Q}$, and $o$ if
$\tape\not\in{Q}$. Note that, as in the standard definition, it is
easy to implement $Q$ as an additional infinite tape that in the
position $\#\tape$ has a \dqt{1} if $\tape\in{Q}$ and a \dqt{0}
otherwise.

Given a set $A$, a representation $\repx$ of $A$, and a set
$N\subseteq{A}$, an $\repx$-oracle for $N$ is a function that, given
the tape $\repx(n)$, produces $\nu$ if $n\in{N}$, and $o$ otherwise;
that is, it is an oracle for $\repx(N)$ inside $\repx(A)$. The TM
$\phi_e^N$ with oracle $N$ is defined in the obvious way.

Given a set $Q\subset\tpset$ we shall indicate with $\comp^Q$ the set
of functions on tapes computable using an oracle for $Q$. Analogously,
given a representation $\repa$ of $A$, and a set $N \subseteq A$, we
shall indicate with $\comp^{N}(\repa)$ the set of functions in
$[A_1\rightarrow{A_2}]$ for some $A_1,A_2\subset{A}$ which are
computable with an $\repx$-oracle for $N$.

\begin{definition} 
  Given two sets $B,C\subseteq{A}$ and a representation $\repx$ of
  $A$, $B$ is \emph{$\repx$-computable in C} if $\chi_B \in
  \comp^{C}(\repx)$.
\end{definition}

Consider a representation $\repx$ of a set $A$ such that $\repx(A)$ is
computationally enumerable. Then, any set $C\subseteq{A}$ is trivially
computable using an $\repx$-oracle for $C$, and as such, it is
computationally enumerable using an $\repx$-oracle for $C$. Use this
enumerator to provide a standard numbering ${\#_\repx^C}$ of $C$. Note
that even if $C$ is always computable using an $\repx$-oracle for $C$,
it need not be computationally enumerable if $\repx(A)$ is not
computationally enumerable.

\begin{definition}
  Let $A, C$ be two infinite countable sets with $C \subset A$. Let $\repx$ be a
  computationally enumerable representation of $A$. The \emph{Turing jump} of $C$ in the
  representation $\repx$ is
  \begin{equation}
    C_\repx^{\prime} = \left\{ c \in C | {\tmach^C}_{{\#_\repx^C}(c)}(\repx(c))\tmstop\right\}
  \end{equation}
\end{definition}

The definition, per se, would not require that $C$ be countable nor
infinite, but the infinite and countable case is the one for
which all the interesting results are derived so, rather than adding
the restriction to the discussions that follow, we have preferred to
make it explicit in the definition.

\begin{theorem}
  $C_\repx^{\prime}$ is not $\repx$-computable in $C$.
\end{theorem}

\begin{proof}
  Let $E^C \in [C \rightarrow C]$ be the enumerator of $C$. Suppose
  $C_\repx^{\prime}$ has a computable characteristic function with
  $\repx$-oracle $C$, ${\chi_{C_\repx^{\prime}}}^C$; define
  $f^C\in[C\rightarrow{C}]$ such that
  \begin{equation}
    {f^C}_\repx(\repx(c)) = 
    \begin{cases}
      {E^C}_\repx({\phi^C}_{{\#_\repx}^C(c)}(\repx(c))) & \mbox{if ${\chi_{C_\repx^{\prime}}}^C(c)=\nu$} \\
      o                       & \mbox{if ${\chi_{C_\repx^{\prime}}}^C(c)=o$}
    \end{cases}
  \end{equation}
  then $f^C\in\comp^C(\repx)$. Thus, we know that there exists
  $c_0\in{C}$ such that for all $c\in{C}$,
  ${\phi^C}_{{\#_\repx}^C(c_0)}(\repx(c))\tmstop$ and
  ${f^C}_\repx(\repx(c))={\phi^C}_{{\#_\repx}^C(c_0)}(\repx(c))$ but,
  for all $c$ such that ${\tmach^C}_{{\#_\repx}^C
    (c)}(\repx(c))\tmstop$,
  ${f^C}_\repx(\repx(c))={E^C}_\repx({\phi^C}_{{\#_\repx}^C
    (c)}(\repx(c)))\ne{\phi^C}_{{\#_\repx}^C(c)}(\repx(c))$.  In
  particular, ${\tmach^C}_{{\#_\repx}^C(c_0)}(\repx(c_0)) \tmstop$, and
  so
  ${f^C}_\repx(\repx(c_0))\ne{\phi^C}_{{\#_\repx}^C(c_0)}(\repx(c_0))$.
  This contradiction proves that $C_\repx^{\prime}$ must be not
  $\repx$-computable in $C$.
\end{proof}

\begin{definition}
  We shall say that $C\le_\repx{B}$ if $C$ is $\repx$-computable in $B$, and that 
  $C\equiv_\repx{B}$ if $C\le_\repx{B}$ and $B\le_\repx{C}$. We define
  \begin{equation}
    \tdg{\repx}{C} = [C]_{\equiv_\repx} = \tset{B | B \equiv_\repx C}
  \end{equation}
  also, we set
  \begin{equation}
    \emptyset_\repx^{(n)} = \tset{a | \tmach_{\#\repx(a)}^{\emptyset_\repx^{(n-1)}}(\repx(a))\tmstop}
  \end{equation}
  where $\emptyset^{(0)}=\emptyset$ is a computable set. Also, set
  \begin{equation}
    \tdgind{0}_\repx^{(n)} = \tdg{\repx}{\emptyset_\repx^{(n)}}
  \end{equation}
\end{definition}

In the following, we shall indicate with $\repc$ the standard
representation of ${\mathbb{N}}$, that in which the number $n$ is
represented as $n+1$ symbols \dqt{1} followed by one \dqt{0}. Then
$\tdgind{0}_\repc=\tdgind{0}$, the degree of standard computable sets,
and
\begin{equation}
  \tdgind{0}_\repc \le_\repc \tdgind{0}_\repc' \le_\repc \tdgind{0}_\repc'' \le_\repc \cdots
  \le_\repc \tdgind{0}_\repc^{(n)} \le_\repc \cdots
\end{equation}
is the standard Turing hierarchy. On the other hand, each
representation $\repx$ induces a hierarchy
\begin{equation}
  \tdgind{0}_\repx \le_\repx \tdgind{0}_\repx' \le_\repx \tdgind{0}_\repx'' \le_\repx \cdots
  \le_\repx \tdgind{0}_\repx^{(n)} \le_\repx \cdots
\end{equation}

\subsection{Representation hierarchy}
Representations and their hierarchies have a connection
with the Turing hierarchy. Let $\chi_\repx^{(n)}$ be the
characteristic function of the set $\emptyset_\repx^{(n)}$, and
$\repc$ the standard representation. Define the class of
representations $\repu^{(k)}:{\mathbb{N}}\rightarrow\tpset$ as
\begin{equation}
  \repu^{(k)}(n) = \bij{\repc(\chi_\repc^{(k)}(n)),\ldots,\repc(\chi_\repc'(n)),\repc(n)}
\end{equation}
We have $\repu^{(k-1)}=\pi_2\circ\repu^{(k)}$, therefore
$\repu^{(k-1)}\wleq\repu^{(k)}$; on the other hand, $\chi_\repc^{(k)}$
is computable in $\repu^{(k)}$, but not in $\repu^{(k-1)}$, so
$\repu^{(k)}\not\seqv\repu^{(k-1)}$. We can consider the equivalence
classes $[\repu^{(k)}]_{\seqv}$. Clearly, if $\repx\seqv\repy$ it is
$\ocomp{\repx}=\ocomp{\repy}$, so the set
$\ocomp{[\repu^{(k)}]_{\seqv}}$ is well defined, and within it are the
set of characteristic functions of sets which are computable by the
class of representations transformationally equivalent to
$\repu^{(k)}$. We call this the \emph{representation degree} of
$\repu^{(k)}$:
\begin{equation}
  \tdgind{\repu^{(k)}} = \rdg{\repu^{(k)}} = \ocomp{[\repu^{(k)}]_{\seqv}}
\end{equation}

\begin{lemma}
  \begin{equation}
    \tdgind{\repu^{(k)}} \subseteq \tdgind{0}^{(k)}
  \end{equation}
\end{lemma}

\begin{proof}
  Let $f\in\tdgind{\repu^{(k)}}$; then there is a TM $\phi$ such that,
  for each $n\in{\mathbb{N}}$,
  $(\repu^{(k)}\circ{f})(n)=(\phi\circ\repu^{(k)})(n)$. Consider a
  tape with the representation $\repc(n)$. Since (trivially)
  $\chi^{(k)}_\repc \in\tdgind{0}^{(k)}$, there is a TM with oracle
  $\emptyset^{(k)}$ that can compute $\chi^{(k)}_\repc$ and, by the
  transitivity of the relation $\le_\repc$, there are TM with oracle $\emptyset^{(k)}$ that can
  compute $\emptyset^{\prime},\ldots,\emptyset^{(k-1)}$; the bijection
  $\bij{,}$ is also computable, therefore there is a TM $\overline{\phi^{\emptyset^{(k)}}}$ with oracle $\emptyset^{(k)}$ that,
  given $\repc(n)$ can compute $\repu^{(k)}(n)$. Applying $\overline{\phi^{\emptyset^{(k)}}}$
  followed by $\phi$ we can compute $f$ with oracle $\emptyset^{(k)}$.
\end{proof}

The following property derives trivially from the observation that
$\chi^{(k)}_\repc\in\tdgind{\repu^{(k)}}$, but $\chi^{(k)}_\repc \not\in\tdgind{0}^{(k-1)}$ 

\begin{lemma}
  \begin{equation}
    \tdgind{\repu^{(k)}} - \tdgind{0}^{(k-1)} \ne \emptyset
  \end{equation}
\end{lemma}

So, there is a hierarchy of representations that in a sense mirrors
the Turing hierarchy. One question that comes naturally is whether
this hierarchy corresponds to an effective increase in computing
power. The following results will allow us to answer this question.

\begin{theorem}
  \label{Aisrc}
  Let $\repx$ and $\repy$ be two representations of
  a set $A$, with $\repx\sleq\repy$. Then, if
  $\repx(A)$ is c.e. so is $\repy(A)$.
\end{theorem}

\begin{proof}
  Let $\repx_{|A}\in[A\rightarrow\repx(A)]$; note that $\repx$ is injective
  and therefore $\repx^{-1}$ exists in $\repx(A)$. The set $\repx(A)$
  is c.e., therefore there are $e'\in[\repx(A)\carrow\repx(A)]$
  and $\tape_0\in\repx(A)$ such that for each $\tape\in\repx(A)$ there
  is $k$ such that $\tape=e'^{k}(\tape_0)$.

  Consider now the function
  $q=\repx^{-1}\circ{e'}\circ\repx_{|A}\in[A\rightarrow{A}]$. It is obvious
  that given $u_0=\repx^{-1}(\tape_0)$ and $u\in{A}$ it is
  $u=q^k(u_0)$ where $k$ is the number such that
  $\repx(u)=e'^{k}(\tau_0)$. That is, $q$ is an enumerator of $A$ and $e' = \frep{q}{\repx}$.

  By a similar argument, it can be seen that
  $e=\repy\circ{q}\circ\repy^{-1} = \frep{q}{\repy}$ is an enumerator of $\repy(A)$ and,
  since $q$ is computable in $\repx$ and $\repx\sleq\repy$, $q$ is also
  computable in $\repy$, that is, $e$ is computable. Therefore $\repy(A)$ is not only c.e., but also with the same
  enumeration function (in $A$) as $\repx(A)$.
\end{proof}

\begin{theorem}
  Let $\repx$ and $\repy$ be two representations of
  a set $A$. If $\repx\sleq\repy$ and $\repx(A)$ is c.e. then
  $\repx\seqv\repy$
\end{theorem}

\begin{proof}
  From theorem \ref{Aisrc} we know that $\repy(A)$ is also c.e. and
  that there is an enumerator $e$ of $A$ that is computable in both
  representations, that is, $e \in\ocomp{\repx} \cap \ocomp{\repy}$.

  Let us begin by showing that $\repy\wleq\repx$; in order to do this,
  we must find a computable function $f$ such that
  $\repy=f\circ\repx$. Let $a\in{A}$, and $\repx(a)$ its
  $\repx$-representation. If $a_0$ is the start value of $e$,
  initialize two tapes with $\repx(a_0)$ and $\repy(a_0)$, then
  simulate the two TM that compute $\frep{e}{\repx}$ and
  $\frep{e}{\repy}$ until we reach an iteration $i$ such that
  $\frep{e}{\repx}^i(\repx(a_0))=\repx(a)$, then $e^i(a_0)=a$, and
  $\frep{e}{\repy}^i(\repy(a_0))=\repy(a)$, that is, on the second tape
  we have the $\repy$-representation of $a$. We therefore have a TM
  that computes $f$.

  In a similar way we can build a TM that computes $g$ such that
  $\repx=g\circ\repy$.
\end{proof}

The following lemma is a direct consequence of the definition of
equivalence by transformability

\begin{lemma}
  Let $\repx$ and $\repy$ be two representations of $A$; if
  $\repx\seqv\repy$ then $\repx\weqv\repy$.
\end{lemma}
 
Note, however, that equivalence is necessary here. It is generally not
true that if $\repx\wleq\repy$ then $\repx\sleq\repy$ nor $\repy \sleq
\repx$. The intuitive idea is that transformability gives us a one way
translation capability, but in order to compute we need to translate
both ways, to represent and to interpret.

From these results, we easily derive the following theorem:

\begin{theorem}
  Let $\repx$ and $\repy$ be two representations of $A$; if $\repx(A)$
  and $\repy(A)$ are c.e.  then either $\repx\seqv\repy$ or
  $\repx\spar\repy$.
\end{theorem}

The proof of all these properties boils down in practice to the
possibility of computing an enumerator of $A$. We can therefore
formulate the previous result in the following guise:

\begin{theorem}
  \label{finalres}
  Let $A$ be a set with an enumerator $e$, and let $\repx$ and
  $\repy$ be two representations of $A$; then if
  $e\in\ocomp{\repx}$ and
  $e\in\ocomp{\repy}$, then
  $\repx\seqv\repy$.
\end{theorem}

The proven results effectively prevent the creation of a hierarchy of
degrees based on representations. Consider the two representations
$\repu$ and $\repu'$: they are representations of ${\mathbb{N}}$,
which is c.e. under $\repu$ so it is not $\repu \sleq \repu'$, and it is not $\repu' \sleq \repu$, since $\chi_{\repc}'$ is
computable in $\repu'$ but not in $\repu$. Therefore, $\repu\spar\repu'$--that is, in order to
gain the possibility of computing $\chi_{\repc}'$, we must give up the
computability of some of the functions that are computable in
$\repu$. Specifically, we lose
the possibility of computing the iterator itself: given
$\repu'(n)=\bij{\repc(\chi_{\repc}'(n)),\repc(n)}$, we can't compute
$\repu'(n+1)=\bij{\repc(\chi_{\repc}'(n+1)),\repc(n+1)}$.

On the other hand, the theorem doesn't tell us anything about the
other degrees of the hierarchy, since none of the range
$\repu^{(k)}({\mathbb{N}})$, $k\ge{1}$ are c.e. and none of them
allows the computation of the iterator. This leads to the idea of
\dqt{relativizing} the properties seen so far through the use of TMs
with oracles. Let $\acomp{R}{\repx}$ be the set of functions that can be
computed in a representation $\repx$ using a TM with oracle $R$. A set
$Q\subseteq\tpset$ is computationally enumerable in $R$ ($R$-c.e.) if
its enumerator can be implemented by a TM with oracle $R$.

\begin{definition}
  A representation $\repy$ is $R$-better than $\repx$
  ($\repx\aleq{R}\repy$) if
  $\acomp{R}{\repx}\subseteq\acomp{R}{\repy}$; the equivalences
  $\repx\aeqv{R}\repy$, $\repx\aseqv{R}\repy$ and the incomparability $\repx\apar{R}\repy$
  are defined in the obvious way.
\end{definition}

The following theorem can be proved in the same way as the preceeding
theorems, simply by replacing all TMs with a TM with the suitable
oracle.

\begin{theorem}
  Let $\repx$ and $\repy$ be two representations of a set $A$ with
  $\repx\aleq{R}\repy$; if $\repx(A)$ is $R$-c.e. then:
  \begin{description}
    \item[i)] $\repy(A)$ is $R$-c.e.;
    \item[ii)] $\repx\aseqv{R}\repy$;
    \item[iii)] $\repx\aeqv{R}\repy$.
  \end{description}
\end{theorem}

From this we derive

\begin{theorem}
  \label{yuppee}
  Let $\repx$ and $\repy$ be two representations of a set $A$, then
  \begin{description}
  \item[i)] if $\repx(A)$ and $\repy(A)$ are $R$-c.e., then either
    $\repx\aseqv{R}\repy$ or $\repx\apar{R}\repy$;
  \item[ii)] if $A$ is R-c.e. with enumerator $e$, $\frep{e}{\repx}\in\acomp{R}{\repx}$ and 
    $\frep{e}{\repy}\in\acomp{R}{\repy}$, then
    $\repx\aseqv{R}\repy$.
  \end{description}
\end{theorem}

Consider now the representations $\repu^{(k)}$ and $\repu^{(k+1)}$;
note that $\chi^{(k+1)}\in\acomp{\emptyset^{k}}{\repu^{(k+1)}}$, but
$\chi^{(k+1)}\not\in\acomp{\emptyset^{k}}{\repu^{(k)}}$, so
$\repu^{(k)}\not\aeqv{\emptyset^{k}}\repu^{(k+1)}$. On the other hand
the enumerator of $\emptyset^{(k)}$ is computable in $\repu^{(k)}$, therefore, it must be
$\repu^{(k)}\apar{\emptyset^{k}}\repu^{(k+1)}$.

\separate

\section{Church-Turing thesis}

The Church-Turing thesis is arguably one of the most famous statements
in computability theory.  The notion of representation in the Turing
machine formalism and its properties, suggest us a new point of view
on the Church-Turing thesis, one that we hope may lead us to a better
understanding of it by seeing it in a somewhat different light.  In a
nutshell, the Church-Turing thesis states that the notion of
computability in any formalism essentially corresponds to that defined
by Turing machines. The thesis' name is due to Turing's work on
providing evidence that this thesis was most likely true. More
precisely, Turing provided a proof in his classic paper
\cite{turing:36} that Turing machines and $\lambda$-calculus (defined
by Church \cite{church:36}) were equivalent.

There are two clarifications to be made vis-a-vis the
Church-Turing thesis and related ideas. On the one hand,
there is a reason for which Church-Turing \emph{thesis} is thus
named. Church-Turing thesis talks about the intuitive notion of
computability, and thus it is not a mathematical statement which can
be proven. One can only provide evidence that suggests that any
reasonable model of computability is equivalent to Turing machines. On
the other hand, the Church-Turing thesis, even if considered true (in its
informal sense), still does not provide an unambiguous definition of
whether a given mathematical function (alternatively, set) is
computable or it is not. The reason for this is that all existing
mathematical models for computation are defined using representative
but specific sets (strings or tapes for Turing machines, lambda
functions for $\lambda$-calculus, natural numbers for (primitive)
recursive functions...). Using these formalisms to consider whether a
function outside the formalism is computable or it is not requires a
process of representation of external elements into the formalism.

This makes our consideration on representations relevant vis-a-vis the
Church-Turing thesis.

\bigskip

\property
\label{AnythingComputable}
\begin{quote}
  For each possible mathematical set, there exists a choice of
  representation which makes that set computable.
\end{quote}

\bigskip

The representation choice is analogue to the one provided in example
\ref{exhalting}. This fact prevents us from giving an absolute
statement about the computability of functions outside of
formalisms. In the case of Turing machines, we can only give absolute
results for functions defined on tapes, in the case of
$\lambda$-calculus for functions on $\lambda$-expressions, and so
on. This makes it impossible to compare, in absolute terms, two
formalisms, as they operate on different sets.

The Turing equivalence theorem, on the other hand, is a formal
comparison of two formalisms, each taking as operating on a different
set.  In light of the apparent ambiguity that representations give to
the concept of computable function, the question then becomes: What
kind of equivalence exists between the Turing machine formalism and
the $\lambda$-calculus formalism? The equivalence provided by Turing's
theorem is \emph{structural}: it shows that the set of
Turing-computable functions (which are functions on strings or tapes)
and the set of computable functions in $\lambda$-calculus (which are
functions on lambda functions) are isomorphic.

In order to provide a full statement of the theorem, then, it would be
necessary to specify what kind of structure they are equivalent
under. This is not explicitly stated in the original Turing theorem,
since the equivalence is proven by showing the somehow obvious
equivalence of the basic functions and using operations such as
composition to build the two isomorphic sets of computable functions
in both formalisms%
\footnote{It could be argued that the reason these precisions were
  not considered by Turing is simply because he was not interested at
  the moment in doing these abstract considerations, but rather was
  more concerned with the computability of certain functions on
  natural numbers with respect to the basic arithmetic operations.}%
. Looking at the details of the proof, the structure for which the
equivalence is proven can be given explicitly. The basic operations
are considered equivalent because they are undistinguishable in
set-theoretic terms. This corresponds to the notions of injectivity
(cardinality) and composability, and derived ones. That is, there is a
bijection between strings or tapes and $\lambda$-functions such that
for each computable function and each subset, the cardinalities of the
images and preimages are preserved (if one string can be obtained from
two possible different strings applying a certain Turing machine, then
the corresponding lambda function must be obtainable from two
different $\lambda$-functions applying the corresponding lambda
function, etc); and at the same time, composability is preserved (if
Turing machines $f$ and $g$ are composed, the resulting Turing machine
corresponds to the $\lambda$-function resulting from composing those
corresponding to $f$ and $g$).

This is not different from any other mathematical structure. For
instance, sets are not group-isomorphic until they are given a group
structure, and so long as cardinalities (set-theoretic structure)
correspond, any two sets can be made group-isomorph by making the
right choice of group structure. Computability is usually only
considered for countable sets, so cardinality is not an issue. These
considerations allow us to restate Property \ref{AnythingComputable}
in the following guise.

\bigskip

\property
\label{AnythingComputableFormalism}
\begin{quote}
  For each possible mathematical set, there exists a choice of a
  computation formalism which is structurally equivalent to Turing
  machines such that the set is computable in such formalism.
\end{quote}

\bigskip

This is somewhat the heart of the issue of representations. A change
of representation and a change of formalism are two faces of the same
coin: computability, the same way as groups, rings, vector spaces or
any other algebraic structure, only can be endowed with after we
have specified what structure we are working in. Thus, the usual
notion of computability of problems of natural numbers corresponds to
a computability structure specified on natural numbers and which
relates to its ring structure in a particular way, namely, in that it
enables the computation of the successor function. Other equivalent
but incompatible structures could be considered on the same set,
providing an equivalent to considering different representations of
natural numbers.

\section{Some Concluding Remarks}
A Turing machine is an abstract device that implements syntactic
functions, that is, functions that transform finite sequences of
symbols on a tape into sequences of symbols on the same tape. We have
shown that using this formalism to define the computability of
functions specified on arbitrary sets requires a process of
representation that cannot, \emph{a priori}, be ignored.  The same
happens with any other model of computability that is defined on a
representative set, such as $\lambda$-functions for
$\lambda$-calculus, recursive functions or natural numbers for
register machines.

The way these formalisms are defined is historically motivated by the
resemblances they exhibit with the mental and physical processes
carried by a mathematician when working on a problem using pen and
paper. This metaphor makes the computability of certain functions and
operators evidently desirable: composition and partial application of
computable functions, computability of the identity function and of a
comparator function. While the computability of an enumerator is
intuitively assumed, it is not derivable from the other properties in
a general computability model, and it is generally not included as an
additional axiom for the model. We have shown that there are
theoretical reasons that compel the inclusion of the computability of
an enumerator as an additional axiomatic capability of common computability
models or, alternatively, of the representations carried when using
these models. Theorem \ref{finalres} of this paper proves that any
computability model with the usually assumed properties can never be
more powerful than another one with the same properties plus
enumerability.

This result, however, gives rise to new questions. The central one is
related to the fact that there are uncountably many enumerators of a
countable set, and that they are obviously not all equivalent. When we
are not thinking of a specific purpose for our computability model,
there is no principled reason to favour the use of one of these
enumerators over others. This creates a whole (uncountable) set of
possible computability models over the same set, all of them mutually
incomparable. As an example, consider two representations of the
cartesian product of Turing machines and input tapes. The first one is
the usual representation used for the classical universal Turing
machine, and it allows the computation of the typical enumerator of
Turing machines and tapes, and thus is a computationally enumerable
representation. The second one allows the computation of another
enumerator which enumerates Turing machines and tapes in a way such
that odd numbers correspond to halting Turing machines and tapes,
while even numbers correspond to non-halting Turing machines and
tapes. Inside each of these subsets, the machines are enumerated using
the enumeration provided in the first representation. As we showed,
the implementation of this enumeration using a representation is
trivial, as it is on the mapping function, outside all formalism,
where we can include all the "magic" necessary. So much that in
practice, the Turing machine implementing both enumerators can
actually be the same one (the general enumerator of tapes). It
  is obvious that these two enumerators are incomparable, since the
  latter allows the computation of the halting problem while the
  former does not, while both being computationally enumerable, and
  the proven theorem can then be applied to prove the
  incomparability.

The kind of problems that we want to solve, or more generally, the
kind of mathematical structure for which we want to create a
computability model (e.g., natural numbers with successor function)
will indicate which enumerator we should use. We can see
  this as thinking that instead of considering the absolute
  computability of arbitrary abstract functions, we can consider their
  co-computability (or co-decidability respectively for set
  membership). For example,
no computability model for natural numbers that allows the computation
of addition also allows the computation of the halting problem (as
typically formulated for natural numbers): these problems are
incompatible. The general Turing jump definition included in the text
generalizes the halting theorem's traditional statement that there are
absolutely incomputable problems over representations to provide a new
statement that in any representation, there are some incomputable
problems.

We have also shown that the structure of the family of models
generated by changes in representations is inherently different from
that provided by oracle machines. While the introcuction of oracles
strictly increases the computation power of a model, a change of
representation is not guaranteed to increase or even retain its
computation power. Moreover, while oracle machines cannot be produced
simply by cleverly exploiting the typical definition of Turing
machines, instead requiring an extension on the definitions,
representations can. We have proven that the "degrees of
representation" originated from using representations are related but
not equal to the classical hierarchy of degrees of recursive
unsolvability provided by oracle machines.

In the section on the Church-Turing thesis we suggest that
computability could be defined as a mathematical structure rather than
as a representative model defined on a representative set, thus
transforming the problem of representation into a problem of
specifying the appropriate structure in the represented set. We
believe that this approach is more in line with the usual work in
mathematics, while at the same time avoiding some ambiguities and
apparent uncertainties present when using models such as Turing
machines or $\lambda$-calculus. We believe that trying to define this
algebraic structure in the same terms as groups, vector fields or
topological spaces is a promising direction for future work.

Finally, a very promising line of future work is the study of the
relation between representation and computational complexity.
For
example, while multiplication of integers represented in unary is an
$O(n^2)$ problem and an $O({log}^2(n))$ problem in binary, if we
represent integers by their decomposition in prime factors, it becomes
an O(1) problem, however making addition a potentially exponential
problem. 
We believe that a generic approach to the problem of the effect of
representation on complexity under an abstract formalism similar to
the one proposed in this text could be useful and possibly provide new
directions for work on usual problems in complexity.


\end{document}